\documentclass[a4paper,12pt,english]{article}
\usepackage{graphicx,url,amsfonts,amsmath,latexsym,amssymb,amsthm,fixmath}
\usepackage{lineno/lineno,hyperref,ifthen}
\usepackage{color,geometry}
\usepackage{subfig}
\usepackage{multirow}
\usepackage{diagbox}

\usepackage[T1]{fontenc}
\usepackage{lmodern}
\usepackage{babel}
\usepackage[babel]{microtype}

\usepackage{hyperref}

\graphicspath{{figures/}}
\def\svgpath{figures/}

\newcommand{\includesvg}[2][]{%
\def\tempa{#1}\def\tempb{}%
\ifx\tempa\tempb\else\let\svgwidth\tempa\fi
\input{\svgpath#2.pdf_tex}
}

\newtheorem{conjecture}{Conjecture}
{\theoremstyle{plain}
\newtheorem*{conj}{Conjecture}}
\newtheorem{theorem}{Theorem}
\newtheorem{lemma}{Lemma}
\newtheorem{remark}{Remark}

\def\Z{\mathbb Z}

\title{Some Triangulated Surfaces without Balanced Splitting\thanks{This work has been partially supported by the LabEx PERSYVAL-Lab (ANR-11-LABX-0025-01)}}

\author{Vincent Despr\'e\thanks{%
Gipsa-lab, CNRS, Grenoble, France;
\url{Vincent.Despre@gipsa-lab.fr}}%
\and Francis Lazarus\thanks{%
Gipsa-lab, CNRS, Grenoble, France;
\url{Francis.Lazarus@gipsa-lab.fr}}%
}

\begin{document}
\maketitle

\begin{abstract}
Let $G$ be the graph of a triangulated surface $\Sigma$ of genus $g\geq 2$. A cycle of $G$ is splitting if it cuts $\Sigma$ into two components, neither of which is homeomorphic to a disk. A splitting cycle has type $k$ if the corresponding components have genera $k$ and $g-k$. 
It was conjectured that $G$ contains a splitting cycle (Barnette '1982). We confirm this conjecture for an infinite family of triangulations by complete graphs but give counter-examples to a stronger conjecture (Mohar and Thomassen '2001) claiming that $G$ should contain splitting cycles of every possible type.
\end{abstract}

\section{Introduction}

A \emph{splitting  cycle} on a topological surface is a simple closed curve that cuts the surface into two non-trivial pieces, none of which is homeomorphic to a disk. See Fig.~\ref{fig:cy}. A torus does not have any splitting cycle but any closed surface (orientable or not) of genus at least two admits a splitting cycle. Given a combinatorial surface, that is a cellular embedding of a graph $G$ into a surface $\Sigma$, it is natural to ask whether $G$ contains a cycle that is a splitting cycle in $\Sigma$. Here, a cycle in a graph is a closed walk without any repeated vertex.
\begin{figure}[h]
\begin{center}
\includegraphics[width=.5\linewidth]{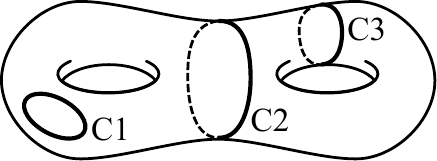}
\caption{A cycle may be null-homotopic (C1) or separating but non null-homotopic (C2) or neither null-homotopic nor separating (C3). C2 is also called a splitting cycle.}
\label{fig:cy}
\end{center}
\end{figure}
It is known to be NP-hard to decide whether a combinatorial surface contains a splitting cycle or not~\cite{art-cabe-find,art-cham-spli}. However, it was conjectured by Barnette that 
\begin{conjecture}[Barnette '1982~\protect{\cite[p. 166]{liv3}}]\label{conj-barn}
Every triangulation of a surface of genus at least 2 has a splitting cycle.
\end{conjecture}

Two splitting cycles on a topological surface have the same \emph{type} if there exists a self-homeomorphism of the surface that maps one cycle to the other one. On an orientable surface of genus $g$ there are $\lfloor g/2\rfloor$ possible types corresponding to splittings into components of respective genus $k$ and $g-k$, $1\leq k \leq g/2$.
A stronger version of the above conjecture was later proposed.
\begin{conjecture}[Mohar and Thomassen '2001~\protect{\cite[p. 167]{liv3}}]\label{conj-mohar}
Every triangulation of an orientable surface of genus at least 2 has a splitting cycle of every possible type. 
\end{conjecture} 

In this article we prove that Conjecture~\ref{conj-barn} holds for the embeddings of the complete graphs $K_{n}$ described by Ringel and Youngs~\cite{liv-ring-map} or by Gross and Tucker~\cite{gt-tgt-87} when $n\equiv 7$ modulo $12$. 
We next present counter-examples to Conjecture~\ref{conj-mohar} that also disprove a stronger conjecture of Zha and Zhao~\cite{art-zha-onno}. Let $M_{19}$ be one of the embeddings of $K_{19}$, the complete graph on 19 vertices, given by Lawrencenko et al.~\cite{lnw-tntos-94}. From the Euler characteristic it is easily seen that $M_{19}$ is a triangulation of  genus $20$. 
A brute force approach to test if any of the cycles of $K_{19}$ is splitting in $M_{19}$ would lead to years of computations. Thanks to a simple branch and cut heuristic we were able to check on a computer that Conjecture~\ref{conj-mohar} fails for  $M_{19}$.  Only 4 of the 10 possible types  occur and, in particular, it is not possible to split $M_{19}$ into two pieces of equal genus. 

After reviewing some terminology and notations we survey the relevant works in  Section~\ref{sec:state}. In Section~\ref{sec:splitting} we provide splitting cycles for the embedding of $K_{12s+7}$ described by Ringel and Youngs~\cite{liv-ring-map}  or Gross and Tucker~\cite{gt-tgt-87}. Our counter-examples to Conjecture~\ref{conj-mohar} are verified thanks to a branch and cut heuristic 
presented in Section~\ref{algo}. The results of our experiments on  $M_{19}$ as well as on other orientable embeddings of bigger complete graphs are discussed in  Section~\ref{res}. 

\section{Backgrounds and terminology}\label{sec:backgrounds}
\paragraph{Combinatorial surfaces} 
A \emph{map} or \emph{combinatorial surface}  is a cellular embedding of a graph in a topological surface. Informally, this is a drawing of a graph on a surface $\Sigma$ such that the edges are drawn as simple non-crossing curves and such that the complement of the graph is a union of topological open disks. If $\Sigma$ is a compact surface with boundary this implies that the boundary of $\Sigma$ is included in the graph. Despite its topological aspect, a map can be encoded by the purely combinatorial data of a rotation system~\cite{liv3}. An \emph{automorphism} of a map is an automorphism of its graph that commutes with the rotation system of the map.
The \emph{genus} $g$ of a map $M$ is the genus of its embedding surface $\Sigma$. Let $V$ and $E$ be the respective number of vertices and edges of the graph $G$ of $M$, and let $F$ be the number of components of the complement of this graph in $\Sigma$.   By Euler's formula, the Euler characteristic $\chi(M):= V-E+F$ is equal to $2-2g$ if $\Sigma$ is orientable and to $2-g$ otherwise.
Another relevant parameter for the existence of a splitting cycle is the \emph{face-width} of $M$; it is the least number of intersections between the embedding of $G$ and any non-contractible cycle on $\Sigma$. 

\paragraph{Triangulations} A map is a \emph{triangulation} if its graph is simple (without loop or  multiple edge) and every face has three sides. We exclude as a triangulation the embedding of a 3-cycle in a sphere. Triangulations are also called \emph{simplicial} triangulations. 
An edge $e$ of a triangulation $T$ can be contracted if the two incident faces are the only 3-cycles to which it belongs and if $T$ is not the embedding of $K_4$ in a sphere. The \emph{contraction} of $e$ results in a triangulation $T'$ obtained by identifying the two vertices of $e$ to a new vertex, deleting $e$ and replacing by single edges the two multiple edges created by the two collapsed triangles. See Figure~\ref{fig:contract}. Note that $T'$ embeds in the same topological surface as $T$. 
A triangulation is \emph{irreducible} when none of its edges can be contracted. Every irreducible triangulation of positive genus has face-width three. Indeed, every edge of such a triangulation belongs to a non-contractible cycle of length three.

\section{State of the art}\label{sec:state}

\subsection{Splitting cycles on Triangulations}\label{subsec:triangle-split}

If $T'$ is obtained from a triangulation $T$ by an edge contraction, then every splitting cycle in $T'$ is the contraction of at least one splitting cycle of the same type in $T$. See Figure~\ref{fig:contract}. 
\begin{figure}[h]
  \centering
\includegraphics[width=.7\linewidth]{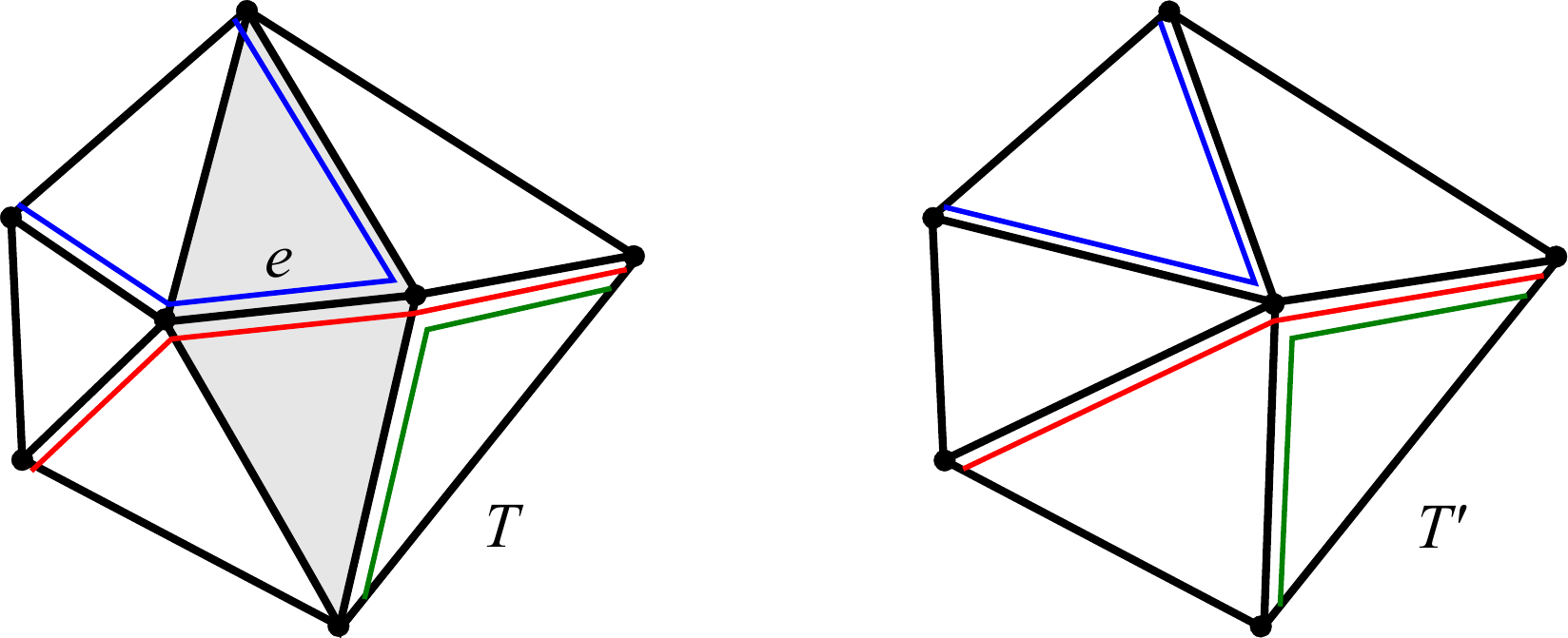}
  \caption{The contraction of edge $e$ in $T$ and the resulting triangulation $T'$. Every simple path on $T'$ is the contraction of (at least) a simple path on $T$.}
  \label{fig:contract}
\end{figure}
It follows that Conjectures~\ref{conj-barn} and~\ref{conj-mohar} can be equivalently restricted to irreducible triangulations. On the other hand, 
Barnette and Edelson~\cite{art-barn-allo,art-barn-all2} proved
that the number of irreducible triangulations of a given surface is finite.  Nakamoto and Ota~\cite{art-nao-note} further showed that the number of vertices in an irreducible triangulation is at most linear in the genus of the
surface. (This result has been extended to surfaces with boundaries by Boulch and al.~\cite{art-bcn-irre}. The best upper bound known to date is due to Joret and Wood~\cite{art-jw-irre} who proved that this number is at most $\max\{13g - 4, 4\}$.) In theory, one can thus list all irreducible triangulations with fixed genus. This makes conjectures~\ref{conj-barn} and~\ref{conj-mohar}  decidable for fixed genus. Indeed, we can consider every irreducible triangulation in turn to test whether one of its cycles is splitting and compute its type as the case may be. Sulanke~\cite{art-sula-gene} describes an algorithm for generating all the irreducible triangulations with given genus and was able to list the irreducible triangulations of the orientable surface of genus 2 and of the non-orientable surfaces up to genus 4. According to Sulanke there are already $396784$ irreducible triangulations of the orientable surface of genus 2 and $6297982$ irreducible triangulations of the non-orientable surface of genus 4.
In practice, the number of irreducible triangulations is growing too fast and the technique cannot be used for higher genera. 
Thanks to its enumeration Sulanke could
conclude by brute force computation that Conjecture~\ref{conj-barn} is true for the orientable surface of genus 2~\cite{art-sula-irre}. A formal and highly technical proof seems to have appeared in Jennings' thesis~\cite{PhDDJ}. Sulanke~\cite{art-sula-irre} also gives a simple counter-example to an extension of Conjecture~\ref{conj-mohar} to the non-orientable case of genus 3. It is constructed from a triangulation of a torus and a triangulation of a projective plane. Remove a triangle from each of those surfaces and glue them along their boundary. Let $C$ be the joining cycle in the resulting non-orientable triangulation of genus 3.  This surface cannot be cut by any cycle  into a perforated Klein bottle and a projective plane since the cycle would  have to cross the length three cycle $C$ at least four times. A similar argument holds when reversing the roles of the torus and the Klein bottle, that is gluing a Klein bottle with a projective plane.
To our knowledge, no progress has been made on conjectures~\ref{conj-barn} and~\ref{conj-mohar} since then.

\subsection{Splitting cycles on maps with large Face-width }
The splitting cycle Conjectures for triangulations have their counterpart for maps with face-width 3:
\begin{conjecture}[Zha and Zhao '1993~\cite{art-zha-onno}]\label{conj-zha1}
Every map of genus at least 2 and face-width at least 3 has a splitting cycle.
\end{conjecture}
\begin{conjecture}[Zha and Zhao '1993~\cite{art-zha-onno}]\label{conj-zha2}
Every map of genus at least 2 and face-width at least 3 has splitting cycles of every possible type.
\end{conjecture}
Since a triangulation has no loop or multiple edge, its face-width is at least 3. The above conjectures are thus stronger than conjectures~\ref{conj-barn} and~\ref{conj-mohar} respectively. It was proved by Zha and Zhao~\cite{art-zha-onno} that a map of genus $\geq 2$ with face-width at least 6 in the orientable case and at least 5 in the non-orientable case has a splitting cycle. Each of their constructions leads to splitting cycles of type $1$, which seems to be the most occurring type.
For maps of genus 2, those conditions were lowered to face-width 4 in the orientable case~\cite{art35} and face-width 3 otherwise~\cite{art-robe-onth}. 
This last case is somehow buried in a paper related to the computation of the genus of a graph. For easy reference, we extract below the result and its proof.
\begin{theorem}[Robertson and Thomas~\cite{art-robe-onth}]
Every non-orientable map of genus 2 and face-width at least 3 has a splitting cycle.
\end{theorem}
In fact, the condition on the face-width can be lowered to the condition that any closed curve homotopic\footnote{Two closed curves are \emph{homotopic} if one can be continuously deformed into the other.}  to some fixed non-separating two-sided simple loop $\ell$
\begin{figure}[h]
  \centering
\includesvg[.85\textwidth]{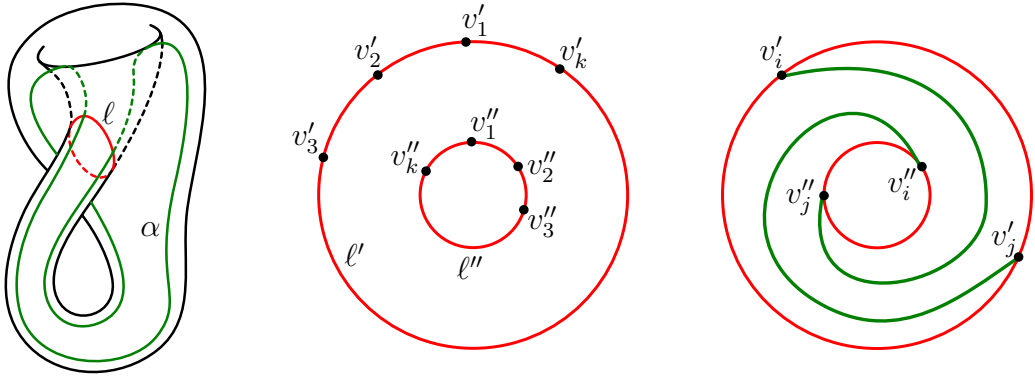}
  \caption{Left, a Klein bottle with a non-separating two-sided curve $\ell$ and a splitting cycle $\alpha$.  Middle, the annulus after cutting along $\ell$. Right, The two paths in the annulus merge to a splitting cycle $\alpha$ on the Klein bottle.}
  \label{fig:KleinBottle}
\end{figure}
intersects the graph of the map at least three times.
\begin{proof}
Consider a map as in the theorem and let $\Sigma$ be the Klein bottle on which its graph is embedded. With a little abuse of notations we write $G$ for the graph as well as its embedding. Choose $\ell$ as above on $\Sigma$ that further minimizes the number $k\geq 3$ of intersections with $G$. After cutting along $\ell$ we get an annulus bounded by two copies $\ell'$ and $\ell''$ of $\ell$ with opposite orientations. See Figure~\ref{fig:KleinBottle}. 
The graph $G$ is also cut by $\ell$, yielding a graph $G'$ with  $k$ vertices $v'_1,\ldots,v'_k$ in cyclic order along $\ell'$ that  correspond to $k$ other vertices $v''_1,\ldots,v''_k$ on $\ell''$.  By the minimal property of $\ell$, every closed curve in the annulus separating $\ell'$ from $\ell''$ cuts $G'$ at least $k$ times. 
By Menger's theorem it easily follows that $G'$ contains $k$ vertex disjoint paths connecting $v'_i$ to $v''_{\sigma(k+1-i)}$ ($i=1\dots k$) for some circular shift $\sigma$ of $[1,k]$. The permutation $i\mapsto \sigma(k+1-i)$ is an involution that cannot be the identity as $k\geq 3$. Any of its 2-cycle $(i,j)$ provides a splitting cycle by merging the path from $v'_i$ to $v''_{j}$ and the path from $v'_{j}$ to $v''_i$. See Figure~\ref{fig:KleinBottle}.
\end{proof}

\subsection{Embeddings of complete graphs}\label{sec:embed}
As noted in Section~\ref{subsec:triangle-split}, one only needs to test conjectures~\ref{conj-barn} and~\ref{conj-mohar} against  irreducible triangulations. Any triangulation whose graph is complete is irreducible. On the other hand, Sulanke's experimentation~\cite{art-sula-gene} on irreducible genus two triangulations suggests that denser graphs, i.e. triangulations with fewer vertices, have longer --- hence potentially fewer --- splitting cycles. It is thus legitimate to confront the conjectures with complete graphs. From the proof of the map color theorem~\cite{liv-ring-map} it is known that for each $n\geq 4$ with $n\equiv 0,3,4$ or 7 modulo 12 the complete graph on $n$ vertices triangulates an  orientable surface (see also~\cite[Sec. 5.1.5]{gt-tgt-87}). A similar result holds for triangulating non-orientable surfaces when $n \not\equiv 2,5$ modulo 6. Describing actual embeddings, i.e., rotation systems, in order to determine the genus of complete graphs was a major achievement of Ringel and Youngs~\cite{liv-ring-map}. Recent studies indicate that the number of non-isomorphic triangulations of a surface by the same complete graph is actually quite high~\cite{art-korz-onth,art-elli-tria,art-gran-2012}. 
Lawrencenko et al.~\cite{lnw-tntos-94} identify three triangular embeddings of the complete graph $K_{19}$
which they prove to be non-isomorphic with the help of a
  computer (see~\cite{art-korz-onth} for a non-computer proof). Each one occurs as a covering of one of the orientable
  genus two \emph{base} maps as shown Figure~\ref{fig:K19Embeddings}. 
    \begin{figure}[h]
  \centering
\includegraphics[width=.85\textwidth]{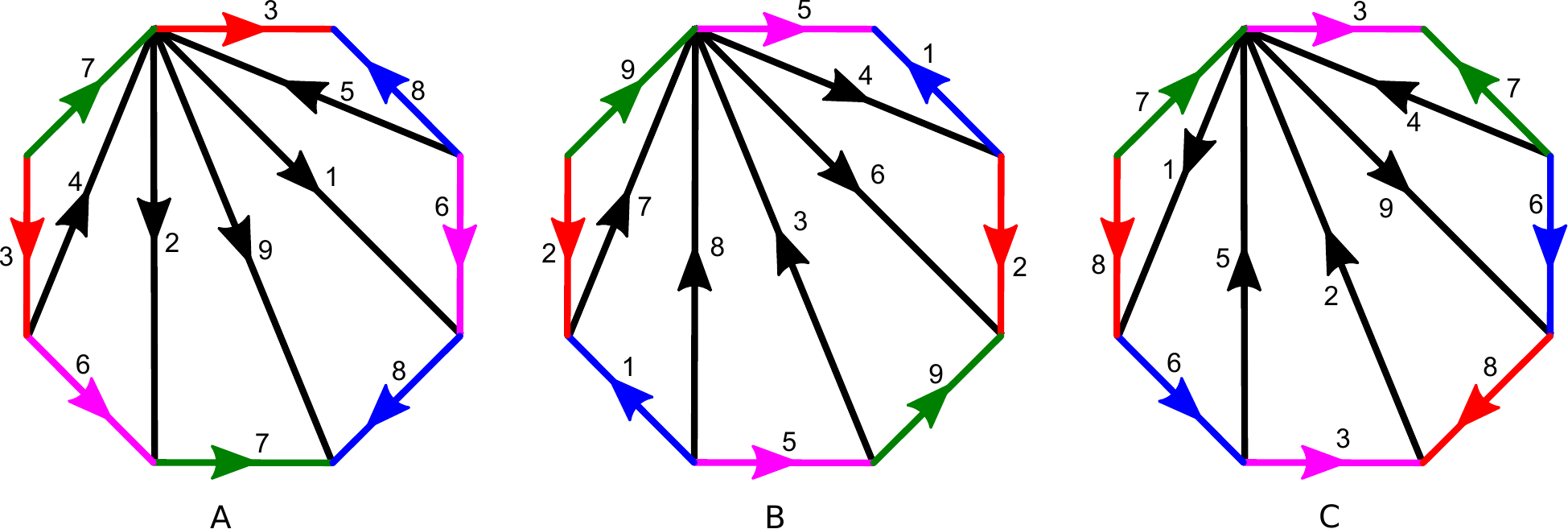}
    \caption{Three genus two maps with voltages in $\Z_{19}$. (The arcs opposite to those represented must receive opposite voltages.) For each octagon the sides should be pairwise identified according to their voltages. This results in each case A, B, C in a map with 9 edges and a single vertex. The corresponding coverings provide triangular embeddings of the complete graph $K_{19}$. Case A is from Ringel and Youngs contruction and case B from Gross and Tucker~\cite{gt-tgt-87}. }
\label{fig:K19Embeddings}
\end{figure}
Those coverings are generated from assignments of the directed edges (or arcs) of each base map to elements of $\Z_{19}$, the cyclic group of order $19$. Two opposite arcs should receive opposite assignments and the sum of the assignments along an oriented facial cycle should be zero, whence the name of \emph{voltage} given to such assignments~\cite[Ch. 2]{gt-tgt-87}. The three voltages in~\cite{lnw-tntos-94} happen to be injective so that the arcs can be identified with their assignment. This leads to the following simple description of each covering. Given one of the base maps with its voltage, we label the vertices of $K_{19}$ with $\Z_{19}$ and declare $(i,j,k)$ to be a triangle iff the three arcs with respective assignments $j-i, k-j$ and  $i-k$ form a facial cycle in the base map. This construction can be generalized to produce embeddings of $K_{12s+7}$ for every positive integer $s$. Gross and Tucker~\cite{gt-tgt-87} use the base map reproduced on Figure~\ref{fig:tenTriangles} as a $4(s+1)$-gon whose sides are pairwised identified according to their voltage. The resulting surface has genus $s+1$ and is covered by a triangular embedding $M_{12s+7}$ of  $K_{12s+7}$ whose genus is $1+s(12s+7)$. As for $M_{19}$, we may identify the vertices of $K_{12s+7}$ with $\Z_{12s+7}$ so that  $(i,j,k)$ is a triangle of $M_{12s+7}$ iff $j-i, k-j$ and  $i-k$ label the three sides of a triangle on the left Figure~\ref{fig:tenTriangles}.

\section{Many splitting cycles}\label{sec:splitting}
 Our implementation for searching splitting cycles lead us to discover that every triangular embedding $M_{12s+7}$ as above has a splitting cycle of type 1. This confirms Conjecture~\ref{conj-barn} for those triangulations. Note that by gluing irreducible triangulations along the boundary of a triangle we may obtain arbitrarily large irreducible triangulations that all possess a splitting cycle. However, the embeddings of complete graphs are not constructed this way and thus provide non-trivial confirmations of Conjecture~\ref{conj-barn}.
For $s\geq 3$, $M_{12s+7}$ has indeed a splitting cycle of length $8$ given by the circular sequence of vertices:
\[ \gamma_s=(0, 5, 2, 9s+8, 6, 1, 4, 5s+6) \]
This cycle bounds a perforated torus made of ten triangles as pictured on the right in
Figure~\ref{fig:tenTriangles}. On this figure the two copies of edge $(2,4)$ should be identified, as well as the two copies of $(0,6)$. 
\begin{figure}[h]
    \centering
\includegraphics[width=\textwidth]{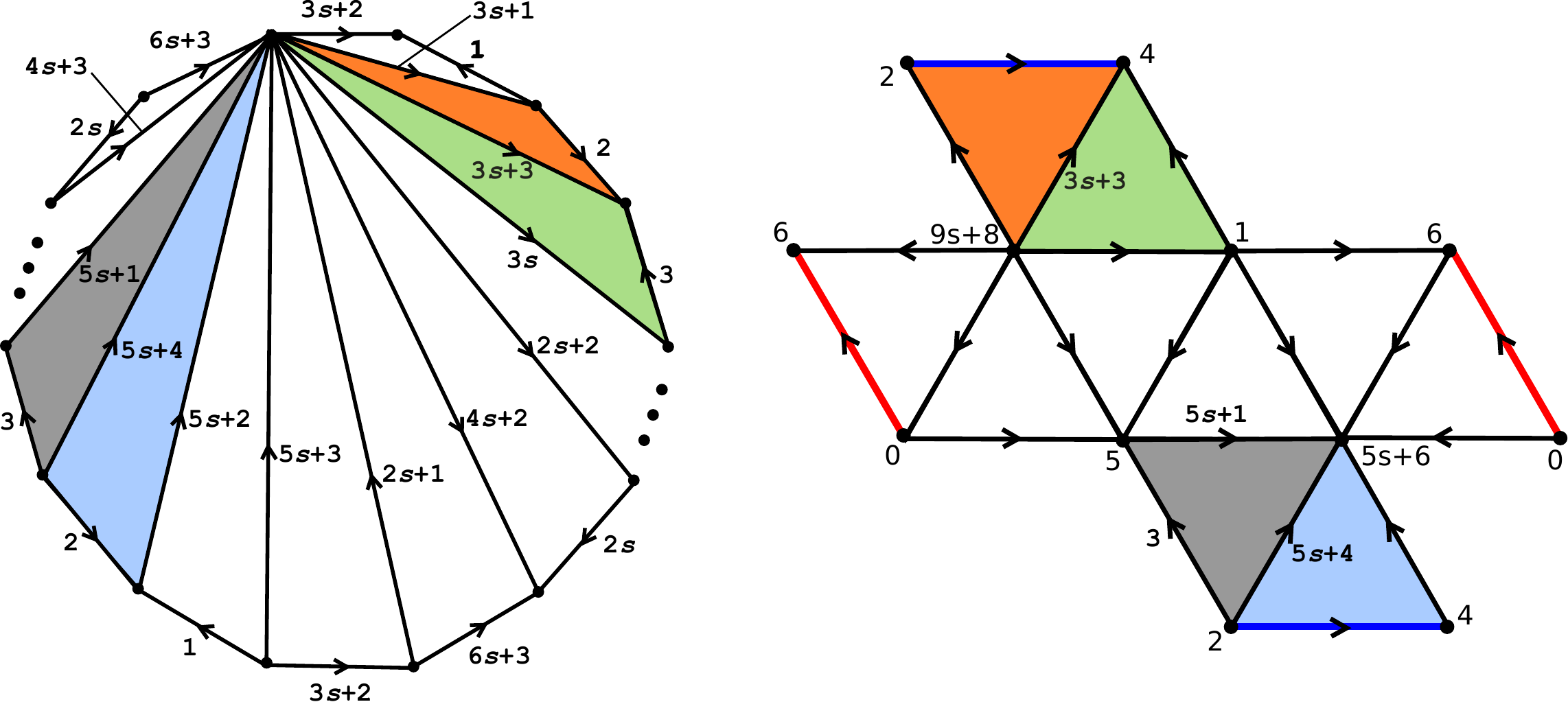}
    \caption{Left, a base map and its voltage for constructing the Gross and Tucker's embedding of $K_{12s+7}$. Note that this base map has a single vertex and is not a simplicial triangulation. Right, the 10 triangles form a sub-surface of genus one with one boundary component in $M_{12s+7}$. Each arc $(u,v)$ is assigned the voltage $v-u$. Hence, the arc $(9s+8,4)$ receives the voltage $3s+3\mod (12s+7)$. Some triangles are colored to show the corresponding covered triangles in the base map.}
\label{fig:tenTriangles}
  \end{figure}
The orbit of $\gamma_s$ by the action of $\Z_{12s+7}$ on $M_{12s+7}$ provides us with $12s+7$ distinct splitting cycles of type 1. Viewing the base map on the left Figure~\ref {fig:tenTriangles} as a fan of $4s+2$ triangles, we can decompose the ten triangles of the perforated torus into two sub-fans of length five in the base map. See Figure~\ref {fig:manySplittings}, left. This construction can be generalized to exhibit many splitting cycles of type 1 by shifting the two fans, leading to the perforated tori as on Figure~\ref {fig:manySplittings}.
\begin{figure}[h]
    \centering
\includegraphics[width=\textwidth]{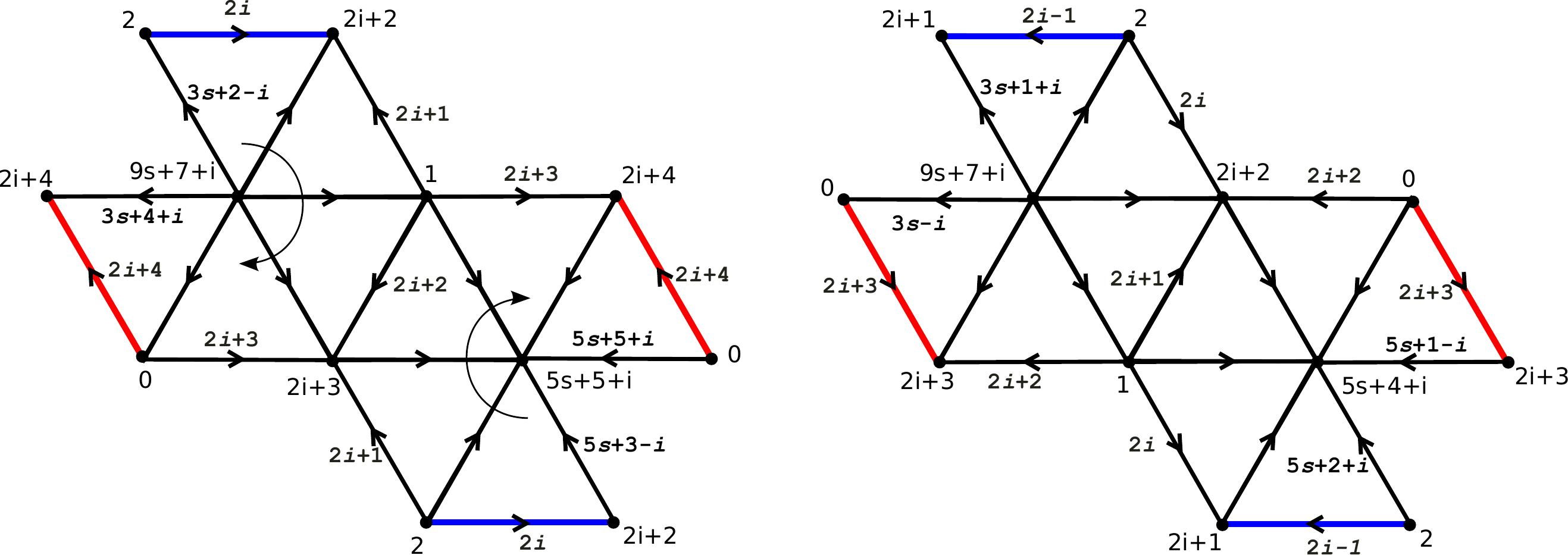}
    \caption{Left, the ten triangles of perforated torus in  $M_{12s+7}$  composed of two sub-fans glued along an edge with even voltage $2i+2$. Right, a similar construction using a gluing edge with odd voltage.}
\label{fig:manySplittings}
  \end{figure}
Each of those cycles can be translated by the action of $\Z_{12s+7}$ to give the following 
$2(s-1)(12s+7)$ distinct splitting cycles of type 1:
\begin{eqnarray*}
  \gamma_{s,i,k} = k+(0, 2i+3, 2, 9s+7+i, 2i+4, 1, 2i+2, 5s+5+i)\\
\gamma'_{s,i,k} = k+(0, 2i+2, 2, 5s+4+i, 2i+3, 1, 2i+1, 9s+7+i)
\end{eqnarray*}
for $i\in[1,s-1]$ and $k\in\Z_{12s+7}$. We can further generalize the construction by gluing larger fans of length $4j+1$ as on Figure~\ref{fig:manySplittings-j} to obtain splitting cycles of type $j$ for $j=1\dots \lceil \frac{s-1}{2}\rceil$.
\begin{figure}[h]
    \centering
\includegraphics[width=.5\textwidth]{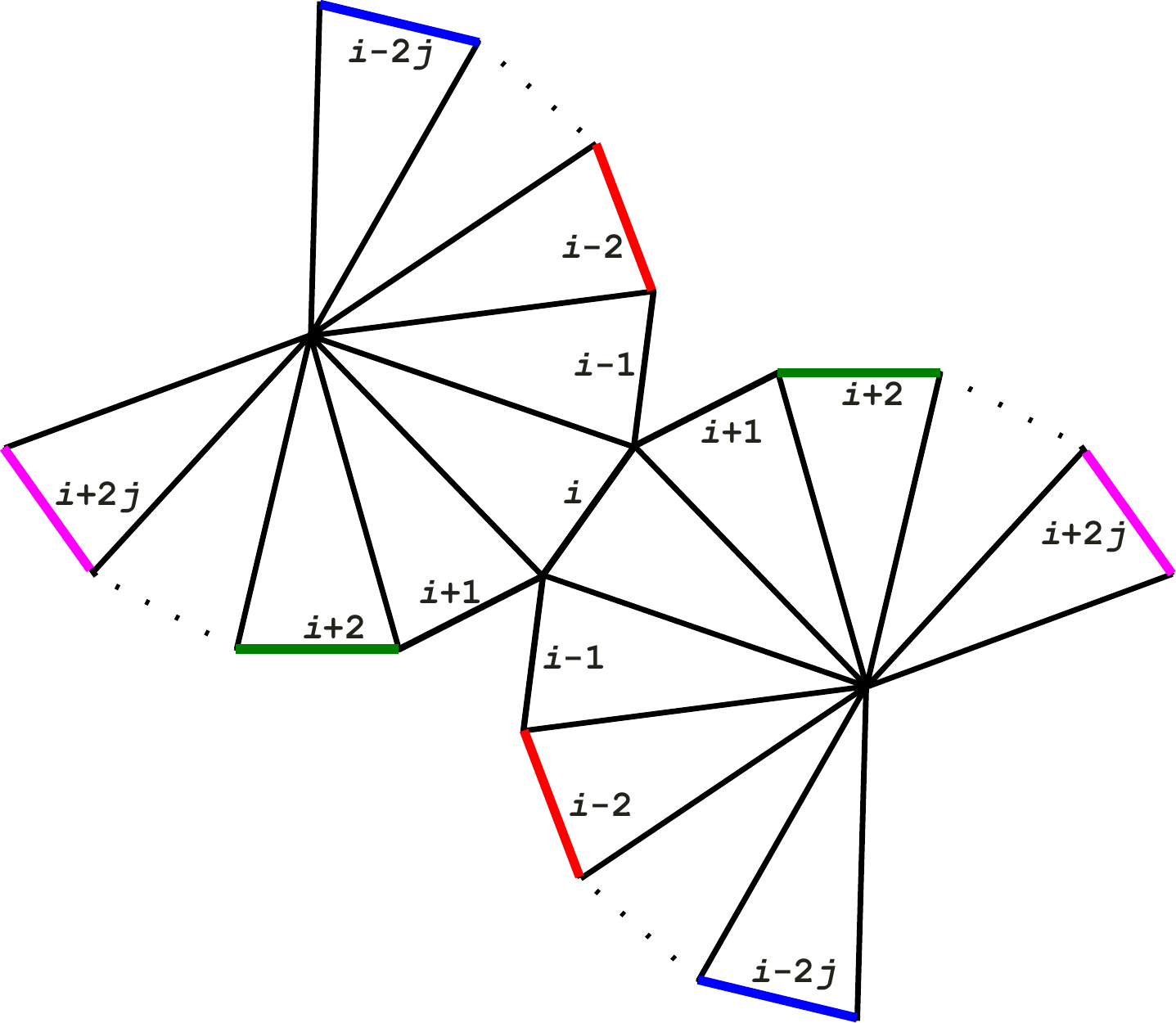}
    \caption{A punctured surface of genus $j$ in $M_{12s+7}$. The thick edges should be glued according to their voltage. }
\label{fig:manySplittings-j}
  \end{figure}
Hence, the embeddings $M_{12s+7}$ constitute an infinite family of irreducible triangulations that admit splitting cycles. 

A similar construction applies to the embeddings of $K_{12s+7}$ by Ringel and Youngs.

\section{Testing algorithm for Conjecture~\ref{conj-mohar}}\label{algo}
Here we provide the details of our implementation for testing Conjecture~\ref{conj-mohar} on any triangular embedding $M_n$ of the complete graph $K_n$ (for a relevant $n$), not necessarily the embedding given by Ringel and Youngs or by Gross an Tucker. 
A straightforward approach consists in checking for every possible cycle in $K_n$  whether it is splitting or not in $M_n$ and computing its type in the former case. This assumes that we can list all the $\sum_{k=3}^n {1\over 2} \binom{n}{k}(k-1)!$ undirected cycles of $K_n$. For $n=19$ this is already more than $9\times 10^{16}$ cycles to test, which is out of reach of current computers. 
\subsection{Pruning the cycle trees}\label{subsec:pruning}
Labelling the vertices of $K_n$ with $\Z_n$, we identify a directed cycle with the sequence of its vertex labels starting with the smallest label. The (directed) cycles can be organized in $n$ rooted trees where the parent of a cycle is obtained by deleting its last vertex. Before exploring those \emph{cycle trees}, we make two simple observations.
\begin{remark}\label{rmk:auto}
  If the automorphism group of $M_n$ acts transitively on the set of its vertices we only need to consider the tree of cycles through vertex $0$. 
\end{remark}
Indeed, an automorphism of $M_n$ does not change the type of a cycle. This remark applies to the three embeddings of $K_{19}$ by Lawrencenko et al.~\cite{lnw-tntos-94} and to all the orientable embeddings of $K_{12s+7}$ by Ringel and Youngs and by Gross and Tucker. In each case $\Z_{12s+7}$ acts transitively by addition on vertex labels. 
\begin{remark}\label{rmk:tr}
  We can assume that a splitting cycle does not contain two consecutive edges bounding a same triangle. We could otherwise replace the two edges by the third one in the triangle. 
\end{remark}
This allows us to decrease by two the degree of the (non root) nodes in the cycle trees. A much more efficient pruning of the cycle trees is provided with the following simple heuristic. Suppose we are given a directed splitting cycle $\gamma$ on $M_n$. We view an edge as a pair of oppositely directed arcs. Color in red or blue all the interior edges of the components of $M_n\setminus \gamma$ respectively to the left or right of $\gamma$. The resulting coloring satisfies that (i) every arc not in $\gamma$ receives the same color as its opposite and (ii) for every vertex, the set of arcs directed inward that vertex is either monochromatic or composed of a red and a blue nonempty sequences separated by two arcs of $\gamma$. This leads to the following coloring test.  As we go down a cycle tree from cycle $(v_0,\dots, v_k)$ to $(v_0,\dots, v_k,v_{k+1})$ we color the arcs pointing outward $v_k$ and distinct from $v_kv_{k-1}$ and from $v_kv_{k+1}$ in red or blue according to whether they lie to the left or right of the directed subpath\footnote{In the non-orientable case, the left and right orientation should be propagated along the path $(v_0,\dots,v_{k+1})$.} $(v_{k-1},v_k,v_{k+1})$. Each time an arc $v_kv$ is colored  we check that
\begin{itemize}
\item $vv_k$ is either colored as $v_kv$ or has not been visited yet, 
\item the arc preceding $v_kv$ around $v$ is not colored with the other color,
\item the arc following $v_kv$ around $v$ is not colored with the other color,
\item if another arc $uv$ has the same color as $v_kv$, then all the arcs entering $v$ with the other color lie on the same side of the path $(v_k,v,u)$.
\end{itemize}
If any of these tests fails the current partial coloring cannot be extended to fulfill the above conditions (i) and (ii). We can thus stop exploring the cycle subtree rooted at $(v_0,\dots, v_{k+1})$. Together with remarks~\ref{rmk:auto} and~\ref{rmk:tr} those simple tests happens to be very effective and to reduce drastically the number of cycles to consider. See Section~\ref{res} for experimental results. 
\subsection{Computing the type of a cycle}
When the cycle $\sigma:=(v_0,\dots, v_{k+1})$ passes the above tests it remains to check if $\sigma$ is splitting and to compute its type. To this end we temporarily color the arcs pointing outward $v_0$  and the arcs pointing outward $v_{k+1}$ in a way similar to that of the other $v_i$, $i=1,\dots,k$. We also perform the above tests and reject the cycle if anyone fails. At this point all the arcs outward the vertices of $\sigma$ have been colored. We say that a vertex of $K_n-\sigma$ is \emph{partially monochromatic} if all its colored inward arcs received the same color.
\begin{lemma}\label{lem:separating}
  $\sigma$ is separating if and only if all the vertices of $K_n-\sigma$ are partially monochromatic and all of the same color.
\end{lemma}
\begin{proof}
  If $\sigma$ is separating then one of the two components of $M_n\setminus \sigma$ has no interior vertex. Otherwise, there would be an interior vertex in each component. However, those vertices would be connected by an edge of the (complete) graph of $M_n$ leading to a contradiction. The direct implication in the Lemma easily follows. For the reverse implication, suppose that every vertex of $K_n-\sigma$ is partially monochromatic and that none of the arcs pointing to $K_n-\sigma$ has color $c$ for some $c\in \{\text{blue, red}\}$. Then, every arc $a$ with color $c$ must connect two vertices of $\sigma$. The sides of the two triangles of $M_n$ incident to $a$ are thus either in $\sigma$ or have color $c$. It ensues that the set of triangles of $M_n$ each of whose sides is either in $\sigma$ or colored with $c$ forms a subsurface of $M_n$ whose boundary is $\sigma$, proving that $\sigma$ is separating.
\end{proof}
When $M_n$ is orientable and $\sigma$ is separating we can directly compute its type. With the notations of the preceding proof this amounts to compute the genus $g'$ of the $c$-colored component of $M_n\setminus \sigma$. Denote by $A_c$ the number of $c$-colored arcs and let $|\sigma|$ be the number of edges of $\sigma$. Using the Euler characteristic and double counting of the edge-triangle incidences, we easily obtain $g'=(A_c - 2|\sigma| +6)/12$ when the $c$-colored component is orientable and $g'=(A_c - 2|\sigma| +6)/6$ otherwise. 

In practice, we maintain the following information for every vertex of $v\in M_n$: whether it belongs to the current cycle, the number of incident blue edges, the number of incident red edges, and a data-structure to store the colored arcs inward $v$. The updating of this information as well as the color tests can be performed in $O(\log n)$  time per newly colored edge using a simple data-structure that allows to find the next or previous colored arc around $v$ and to insert or remove an arc in $O(\log n)$ time. Traversing the cycle trees in depth first  order our algorithm thus spends $O(n\log n)$ time per cycle in the cycle trees. When $M_n$ is non-orientable we can compute the Euler characteristic of the $c$-colored component in the same amount of time. If the $c$-colored component is orientable the other component must be non-orientable and we are done. Otherwise we may need $O(|K_n|)=O(n^2)$ time to determine the orientability of the other component.

\section{Results}\label{res}
In this section we detail the results from testing our algorithm on embeddings of complete graphs, thus disproving conjectures~\ref{conj-mohar} and~\ref{conj-zha2}. Remark~\ref{rmk:auto} applies to each of the tested embeddings so that we only need to explore the cycle tree rooted at vertex $0$.

\subsection{Embeddings of $K_{19}$}
Our smallest counter-examples to Conjecture~\ref{conj-mohar} are provided by  the three embeddings of $K_{19}$ described in Section~\ref{sec:embed}. We refer to them as A,B and C in accordance with Figure~\ref{fig:K19Embeddings}. A splitting cycle is said to have type $k$ if it cuts the surface into components of respective genus $k$ and $20-k$. The next table shows for each embedding and each type the number NSC of splitting cycles of that type found as we traverse the cycle tree with root vertex 0 and the minimum length of any of those NSC cycles. Note that this minimum length would be the same if we would not take Remark~\ref{rmk:tr} into account for pruning the cycle tree. For instance, we may note that every splitting cycle of type 4 in embedding B is Hamiltonian. 
\begin{center}
\begin{tabular}{c|c|c|c|c|c|c|}
\cline{2-7}
                                        &  Type &  1 &  2 &  3 &  4 &  5-10  \\ 
\hline
\multicolumn{1}{|c|}{\multirow{2}{*}{A}} &  NSC &  450 & 545 & 79 & 18 & 0  \\ \cline{2-7} 
\multicolumn{1}{|c|}{}  &   Min Length &  11 & 14 & 16 & 18 & $\perp$  \\ 
\hline
\multicolumn{1}{|l|}{\multirow{2}{*}{B}} &   NSC &  468 & 494 & 130 & 19 & 0  \\ \cline{2-7} 
\multicolumn{1}{|l|}{}          &  Min Length &  10 & 14 & 18 & 19 & $\perp$  \\ \hline        
\multicolumn{1}{|l|}{\multirow{2}{*}{C}} &   NSC &  355 & 257 & 17 & 36 & 0 \\ \cline{2-7} 
\multicolumn{1}{|l|}{}          &  Min Length &  11 & 15 & 17 & 18 & $\perp$   \\ \hline        
\end{tabular}
\end{center}
None of the three embeddings admits a splitting cycle of type 5 or more, thus disproving Conjecture~\ref{conj-mohar}. In particular, these embeddings cannot be split in a balanced way into two punctured surfaces of genus 10. As a side remark, the fact that the numbers in the table are distinct for the three embeddings is another confirmation that they are not isomorphic~\cite{lnw-tntos-94}. The next table indicates the proportion of contractible and splitting directed cycles among the visited nodes (that passes the color tests of Section~\ref{subsec:pruning}) in the cycle tree. Since every cycle appears with both directions in the cycle tree, the number of splitting cycles is twice the sum in the corresponding NSC row in the previous table. 

\begin{center}
\begin{tabular}{|c|c|c|c|}
\cline{2-4}
\multicolumn{1}{c|}{} &  \# visited nodes &  \# contractible &  \# splitting \\ 
\hline
A & 250221 & 36 & 2164 \\ 
\hline 
B &  244229 & 36 & 2222 \\ 
\hline 
C &  210808 & 36 & 1330 \\ 
 \hline        
\end{tabular}
\end{center}
Hence, thanks to our pruning heuristic, less than $3\times 10^5$ of the $1.8\times10^{17}$ directed cycles of $K_{19}$ are visited. We remark that by a Dehn type argument in the universal cover of the triangulation (see~\cite[Sec. 6.1.3]{liv4}) a contractible cycle that does not contain two consecutive edges of any triangle must be the link (i.e. the boundary of the star) of a vertex. Since every cycle in the cycle tree must contain vertex $0$ this leaves 18 link cycles as contractible cycles and explains the 36 found in each row of the previous table.

\subsection{More counter-examples} \label{subsec:more-counter-samples}
We ran our test algorithm on other embeddings of complete graphs on a quad core laptop with 8Gb RAM. Table~\ref{tab:RY} summarizes our results for the Ringel and Youngs embeddings of $K_n$ with $n\in\{15, 19, 27, 28, 31, 39, 40, 43\}$. Each entry in the table gives the length of the smallest splitting cycle of a given type for some $K_n$, if any. The last row indicates the largest possible type of a splitting cycle in the Ringel and Youngs embedding of $K_n$. It took less than 10 seconds to explore the pruned cycle tree for $K_{19}$, less than one hour for $K_{31}$ and about half a day for $K_{43}$. 
\begin{table}[h]
  \centering
\begin{tabular}{|c|c|c|c|c|c|c|c|c|c|c|}
\hline 
 \diagbox{  Type}{$K_n$} & \textbf{$K_{15}$} & \textbf{$K_{19}$} & \textbf{$K_{27}$} & \textbf{$K_{28}$} & \textbf{$K_{31}$} & \textbf{$K_{39}$} & \textbf{$K_{40}$} & \textbf{$K_{43}$} \\ 
\hline 
\textbf{1} & 8 & 11 & 12 & 12 & 8  & 12 & 10 & 8 \\ 
\hline 
\textbf{2} & 11 & 14 & 16 & 17 & 13  & 15 & 15 & 11 \\ 
\hline 
\textbf{3} & 12 & 16 & 19 & {18} & 15 & 20 & 18 & 12  \\ 
\hline 
\textbf{4} & 13 & 18 & 20 & $\perp$ & 17 &  24 & 19 & 15  \\ 
\hline 
\textbf{5} & 14 & $\perp$ & 27 & $\perp$ & 20 & 26 & 24 & 18\\ 
\hline 
\textbf{6} &  & $\perp$ & $\perp$ & $\perp$ & 21  & 30 & 26 & 20 \\ 
\hline 
\textbf{7} &  & $\perp$ & $\perp$ & $\perp$ & 23  & 32 & 28 & 21 \\ 
\hline 
\textbf{8} &  & $\perp$ & $\perp$ & $\perp$ & 24  & $\perp$ & 30 & 23  \\ 
\hline 
\textbf{9} &  & $\perp$ & $\perp$ & $\perp$ & 28  & $\perp$ & 33 & 24\\ 
\hline 
\textbf{10} &  & $\perp$ & $\perp$ & $\perp$ & 28  & $\perp$ & 35 & 25\\ 
\hline 
\textbf{11} &  &  & $\perp$ & $\perp$ & 29  & $\perp$ & 36 & 27 \\ 
\hline 
\textbf{12} &  &  & $\perp$ & $\perp$ & $\perp$  & $\perp$ & 38 & 29\\ 
\hline 
\textbf{13} &  &  & $\perp$ & $\perp$ & $\perp$ & $\perp$ & 40 & 30 \\ 
\hline 
\textbf{14} &  &  & $\perp$ & $\perp$ & $\perp$ & $\perp$ & $\perp$ & 31 \\ 
\hline 
\textbf{\vdots} &  &  & $\perp$ & $\perp$ & $\perp$ & $\perp$ & $\perp$ & \vdots \\ 
\hline 
\textbf{29} &  &  & & & $\perp$ & $\perp$ & $\perp$ & 42 \\ 
\hline 
\textbf{30} &  &  & & & $\perp$ & $\perp$ & $\perp$ & $\perp$ \\ 
\hline 
\textbf{max type} & 5 & 10 & 23 & 25 & 31 & 52 & 55 & 65 \\
\hline
\end{tabular} 
  \caption{Minimal size of splitting cycles of Ringel and Youngs embeddings according to their type.}
  \label{tab:RY}
\end{table}
We have also tested some of the Gross and Tucker embeddings of $K_{12s+7}$. Some results are listed in Table~\ref{tab:GT}.
\begin{table}[h]
  \centering
\begin{tabular}{|c|c|c|c|}
\hline 
\diagbox{    Type}{$K_n$} & $K_{19}$ & $K_{31}$ & $K_{43}$ \\ 
\hline 
\textbf{1} & 10 & 8  & 8 \\ 
\hline 
\textbf{2} & 14 & 13  & 11 \\ 
\hline 
\textbf{3} & 18 & 17 & 12  \\ 
\hline 
\textbf{4} & 19 & 19 & 15  \\ 
\hline 
\textbf{5} & $\perp$ & 20 & 18\\ 
\hline 
\textbf{6} &  $\perp$ & 25  & 23 \\ 
\hline 
\textbf{7} &  $\perp$ & 26  & 26 \\ 
\hline 
\textbf{8} &  $\perp$ & 26  & 26  \\ 
\hline 
\textbf{9} &  $\perp$  & $\perp$ & 28\\ 
\hline 
\textbf{10} &  $\perp$  & $\perp$ & 34\\ 
\hline 
\textbf{11} &  & $\perp$ & 34 \\ 
\hline 
\textbf{12} &  & $\perp$ & 35\\ 
\hline 
\textbf{13} &  & $\perp$ & 35 \\ 
\hline 
\textbf{14} &  & $\perp$ & $\perp$ \\ 
\hline 
\textbf{max} & 10 & 31 & 65 \\
\hline
\end{tabular} 
  \caption{Minimal size of splitting cycles of Gross and Tucker embeddings according to their type.}
  \label{tab:GT}
\end{table}

\section*{Conclusion}
Our counter-examples to Conjecture~\ref{conj-mohar} were checked with the help of a computer. Can we give a formal proof that would not recourse to a computer (at least for $M_{19}$)? Although we could not find such a proof, three points seem relevant to this purpose. 
\begin{itemize}
\item Suppose there is a splitting cycle of type 10 with $k$ edges in $M_{19}$. As noted in the proof of Lemma~\ref{lem:separating}, on one side of the cut surface we have no interior vertex. Let $f$ and $e$ be the number of faces and edges on this side (the number of vertices is $k$). By double counting incidences we have $3f=2e-k$, and by Euler's formula: $k - e + f = 2-20 -1=-19$. It ensues that $e=2k+57$. Because the graph is simple we also have $e\leq \binom{k}{2}$. This implies $k^2 - 5k -114 > 0$ and in turn $14\leq k\leq 19$. A similar computation shows that a splitting cycle of type half the genus of $M_{12s+7}$ has length at least $(5 +\sqrt{(1+(24s+7)^2)/2})/2$.
\item Every splitting cycle leads to  an arrangement of $12s+7$ splitting cycles thanks to the action of $\Z_{12s+7}$ on $M_{12s+7}$. Is it possible to take advantage of this arrangement to obtain an ``impossible'' dissection of $M_{12s+7}$ when assuming the existence of a splitting cycle of type $\lfloor g(M_{12s+7})/2\rfloor$?
\item Finally, we saw in Section~\ref{sec:splitting} that $M_{12s+7}$ contains many short splitting cycles. On the other hand, the first above point tells that a splitting cycle $\gamma$ of type half the genus is relatively long, hence must cut many of the short splitting cycles. Being separating, $\gamma$ has to cut every other cycles an even number of times. Would this enforce $\gamma$ to have length larger than $12s+7$, leading to a contradiction?
\end{itemize}

Our counter-examples show that it is not always possible to split a genus $g$ triangulation into two genus $g/2$ triangulations. Looking at the tables in Section~\ref{subsec:more-counter-samples}, it seems that the proportion of the types of the splitting cycles of $M_{12s+7}$ is roughly decreasing as $s$ grows. This leads to the following conjecture in replacement of Conjecture~\ref{conj-mohar}.
\begin{conj}
For any positive real number $\alpha \leq 1/2$, there exists a triangulation (or a graph embedding of face-width at least 3) of arbitrarily large genus $g$ that has no splitting cycle of type larger than $\alpha g$.
\end{conj}

\newcommand{\etalchar}[1]{$^{#1}$}

\end{document}